\documentclass[journal,draftclsnofoot,12pt,onecolumn]{IEEEtran}
\usepackage{color}
\usepackage{amsfonts}
\usepackage{verbatim}
\usepackage{multirow}
\usepackage{wrapfig}

\usepackage{subfigure}
\usepackage{cite}      

\usepackage{graphicx}  
\usepackage{amsbsy}
\usepackage{amssymb}
\usepackage{algorithmic}
\usepackage{amsmath}   
\usepackage{prettyref}
\newrefformat{s}{Section~\ref{#1}}

\IEEEoverridecommandlockouts

\usepackage{setspace}
\begin{document}

\title{Downlink Noncoherent Cooperation without Transmitter Phase Alignment}
\author{\authorblockN{Mingguang Xu, Dongning Guo, and Michael L.
    Honig}
  \authorblockA{\\Department of Electrical Engineering and Computer Science\\
  Northwestern University\\
    2145 Sheridan Road, Evanston, IL 60208 USA}
    \thanks{This work was
    supported by the NSF under grant CCF-0644344, the Army Research
    Office under grant W911NF-06-1-0339, DARPA under grant
    W911NF-07-1-0028, and a gift from Huawei. This paper was
    presented in part at the 2010 IEEE Global Communications Conference, Miami, Florida, USA,  December, 2010.}  }


%


\maketitle

\begin{abstract}
Multicell joint processing can mitigate inter-cell interference and
thereby increase the spectral efficiency of cellular systems. Most
previous work has assumed phase-aligned (coherent) transmissions
from different base transceiver stations (BTSs), which is difficult to achieve
in practice. In this work, a {\em noncoherent} cooperative
transmission scheme for the downlink is studied, which does not
require phase alignment. The focus is on jointly serving two users in adjacent cells sharing the same resource block.  The two BTSs partially share their messages through
a backhaul link, and each BTS transmits a superposition of two
codewords, one for each receiver. Each receiver decodes its own message, and
treats the signals for the other receiver as background noise. With
narrowband transmissions the achievable rate region and maximum
achievable weighted sum rate are characterized by optimizing the
power allocation (and the beamforming vectors in the case of multiple
transmit antennas) at each BTS between its two codewords. For a wideband
(multicarrier) system, a dual formulation of the optimal
power allocation problem across sub-carriers is presented, which can be
efficiently solved by numerical methods. Results show that the proposed
cooperation scheme can improve the sum rate substantially in the low
to moderate signal-to-noise ratio (SNR) range.
\end{abstract}


%

\section{Introduction}
\label{s:int}

In cellular networks where each base transceiver station (BTS)
independently transmits to mobile stations
within  its own cell, inter-cell interference is a major limitation on
the sum spectral efficiency. Rather than treating inter-cell
interference as noise, the modern view is that it can be exploited
by coordinating transmissions from the BTSs.
It is well-known that coordinated transmissions can potentially
increase the spectral efficiency dramatically (e.g., see
\cite{GesHan10JSAC,GjeGes08TWC,GesKou11IT,ChoAnd08TWC,LarJor08JSAC,Wyner94IT,
JafFos04EURASIP,FosHua05CISS,KarFos06TWC,AndCho07TWC,
SomZai07IT,MarFet11TWC,JinTse08EURASIP,SimSom08EURASIP,ZhaMeh08TWC,
CaiRam08Allerton,MarFet09GC,SimSom09EURASIP,BenHua09ICC}).
A recent comprehensive review of multicell coordination techniques
is given in \cite{GesHan10JSAC} and references therein.

There can be different levels of BTS coordination.
The basic level is to share channel state information (CSI)
for the direct and interfering channels among the BTSs.
That allows the BTSs to adapt their transmission strategies
to channel conditions jointly, and
includes inter-cell joint power control, user scheduling,
and beamforming \cite{GjeGes08TWC,GesKou11IT,ChoAnd08TWC}.
(See also \cite{ChaSez07Allerton,JorLar08TSP,LarJor08JSAC},
which consider power allocation and beamforming
for peer-to-peer (interference) networks.)
These techniques treat the inter-cell interference as noise,
but it is mitigated by exploiting the heterogeneity of CSI
across different users.

A higher level of coordination is multicell joint processing, which
requires the cooperating BTSs to exchange message data in addition
to CSI \cite{Wyner94IT,JafFos04EURASIP,FosHua05CISS,KarFos06TWC,AndCho07TWC,
SomZai07IT,MarFet11TWC,JinTse08EURASIP,SimSom08EURASIP,
ZhaMeh08TWC,CaiRam08Allerton,MarFet09GC,SimSom09EURASIP,BenHua09ICC}.
Interference can be mitigated by using ``virtual''
or ``network'' multiple-input multiple-output (MIMO)
techniques \cite{Wyner94IT,SomZai07IT,FosHua05CISS,KarFos06TWC,JinTse08EURASIP,CaiRam08Allerton,MarFet09GC, MarFet11TWC},
which view all interfering signals as carrying
useful information. 
Although multicell joint processing can potentially provide
substantial performance gains, it introduces a number of challenges.
In particular, most coordinated transmission schemes in the literature
not only require knowledge of codebooks and perfect CSI at all
transmitters and receivers, but also require the cooperating
transmissions to be aligned in phase so that transmissions superpose
coherently at the receivers. Phase-aligning oscillators at different
geographical locations is difficult, since small carrier
frequency offsets translate to large baseband phase
rotations~\cite{JunWir08ISWCS,MudBro09CM}.

This paper presents a {\em noncoherent} scheme for downlink
cooperation, which does not require phase alignment at the
transmitters. For simplicity, we consider a scenario where two BTSs cooperatively
transmit to two mobiles assigned the same time-frequency resource
block, one in each cell as depicted in Fig.~\ref{f:twocells}. It is assumed the two BTSs
partially or fully share their messages via a
bi-directional dedicated link. Each BTS may transmit a
superposition of two codewords, one for each receiver. Each receiver
decodes only its own message, and treats the undesired signals as
background noise. Assuming that Gaussian codebooks are used to
encode all messages, the proposed scheme is simple:
The message intended for each receiver is in general split into two
pieces to be transmitted by the two BTSs, respectively. The rate and
power allocations across the messages at each BTS are then
optimized. That is, for a given set of channel gains the available
power at each BTS is divided between a signal used to transmit its
own message and a signal used to transmit the message
from the other BTS.

This cooperative scheme is motivated by scenarios where each BTS
has no {\em a priori} information about the phase at the other BTS.
While the optimal (capacity-achieving) cooperative transmission
scheme is unknown, and seems to be difficult to determine, the
proposed rate-splitting scheme is a likely candidate.
Furthermore, it serves as a baseline for comparisons with
other schemes in which limited phase information may be obtained.

\begin{figure}
  \center
  \includegraphics[width=0.5\columnwidth]{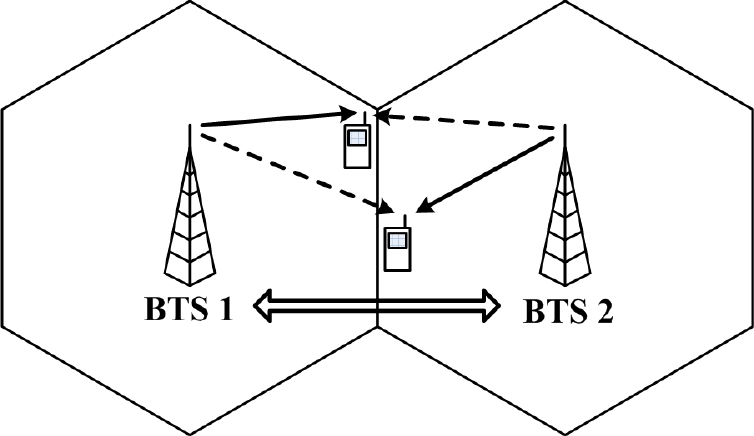}
  \caption{A scenario with two base transceiver stations and two mobiles.}
  \label{f:twocells}
\end{figure}

We optimize the powers allocated across the data streams
and associated beamformers with multiple transmit antennas with
cooperative transmissions for both narrowband and wideband scenarios.
For narrowband channels with a single transmit antenna
the frontier of the achievable rate region is computed by solving
a {\em linear-fractional program}. The weighted sum rate can be maximized
by comparing at most six extremal rate pairs in the constraint
set for transmit power. With multiple transmit antennas, the achievable
rate region can be characterized by maximizing the weighted sum rate
over the allocated power and the beamforming vector for each message,
and the resulting optimization problem can be solved efficiently by
numerical methods. This noncoherent cooperative scheme often achieves
a significantly larger rate region and much higher sum rate
than non-cooperative schemes.

With wideband (frequency-selective) channels, the power is allocated
over multiple sub-carriers. Maximizing the sum rate is in general a
non-convex problem. Under mild assumptions, however, the dual problem
can be solved efficiently. Moreover, we propose a suboptimal power allocation scheme for the case of a single transmit antenna, which admits
a simple analytical solution. This suboptimal scheme performs almost
as well as the optimal power allocation when the direct- and cross-channel
gains are of the same order.

The optimization problems presented can be easily extended
to more than two BTSs and two mobiles.
However, the structure of the solution becomes
more complicated, necessitating general numerical (convex programming)
techniques. Here we focus on the scenario with two mobiles in
adjacent cells since in practice a particular mobile is likely
to have only one relatively strong interferer,
and coordinating among more than two mobiles
across cells becomes quite complicated.
(This complication can be compounded by the scheduler, which may reassign
nearby mobiles different time-frequency resources over successive frames.)
Finally, the two-mobile scenario provides significant
insight into the potential gains of the cooperative scheme.

\section{Narrowband Cooperation Model}
\label{s:narrow}
Consider downlink transmission in two adjacent cells each with the same
set of narrowband channels. Within each cell, the signals from the BTS to its associated
different mobiles occupy non-overlapping time-frequency resources; however,
in each time-frequency slot, there can be inter-cell interference. Here we
consider two mobiles in adjacent cells assigned the same narrowband channel.
Assuming a narrowband system with block fading and single transmit antenna
at each BTS and single receive antenna at each mobile, the baseband signal
received by mobile $j\!=\!1,2$ during the $n$-th symbol interval is
\begin{equation}
y_j{(n)}=\sqrt{g_{j1}}e^{i\theta_{j1}(n)}x_{1}{(n)}+\sqrt{g_{j2}}e^{i\theta_{j2}(n)}x_{2}{(n)}+z_j{(n)}\label{eq:recsig1},
\end{equation}
where $i^2\!=\!-1$, $g_{jk}$ denotes the positive block fading gain from BTS $k$ to mobile $j$, $x_k{(n)}$, for $k\!=\!1,2$, denotes the transmitted signal from
BTS $k$ at the $n$-th symbol interval, $\theta_{jk}(n)$ denotes the phase of the fading channel from BTS $k$ to mobile $j$, and $z_j{(n)}$ denotes the noise at mobile
$j$, which is a sample from a sequence of independent, unit-variance
circularly symmetric complex Gaussian (CSCG) random variables.
\begin{table}
  \caption{Knowledge of CSI at each terminal.}
  \begin{center}
  \label{tCSI}
    \begin{tabular}{c|cccc|cccc}
            & $g_{11}$ & $g_{21}$ & $g_{12}$ & $g_{22}$ & $\theta_{11}(n)$ & $\theta_{21}(n)$ & $\theta_{12}(n)$ & $\theta_{22}(n)$\\
      \hline
      BTS 1 & Y        & Y        & Y        & Y        &             &            & N             & N            \\
      \hline
      BTS 2 & Y        & Y        & Y        & Y        & N             & N             &             &          \\
      \hline
      mobile 1 & Y        &         & Y        &        & Y             &             & Y             &            \\
      \hline
      mobile 2 &         & Y        &        & Y        &             & Y             &             & Y            \\
      \hline
    \end{tabular}
  \end{center}
\end{table}

It is important to specify what channel state information is known to which
transmitters and/or receivers.  The block fading gains
$(g_{11},g_{21},g_{12},g_{22})$ are known to both BTSs.  The gains $(g_{j1},g_{j2})$ are known
to the corresponding receiver $j$.  Usually, these gains are measured by
the receiver and sent back to the transmitters through some feedback link.
Whether $(g_{j1},g_{j2})$ are known to the other receiver is inconsequential
in this study.  The phases $\big(\theta_{j1}(n),\theta_{j2}(n)\big)$ are known or can be acquired by mobile $j$.
Phases $\big(\theta_{11}(n),\theta_{21}(n)\big)$ from BTS 1 are unknown to BTS 2. Likewise, phases $\big(\theta_{12}(n),\theta_{22}(n)\big)$ are unknown to BTS 1. In fact, due to the frequency offset between the two oscillators at the two BTSs, the phase difference $\theta_{j1}(n)-\theta_{j2}(n)$ varies rapidly with $n$ at each mobile $j$. This prohibits one BTS to track the phases originating from the other BTS. Hence coherent combining at the receiver is not feasible. Since receiver $j$ can compensate for phase $\theta_{jj}(n)$, it can be assumed without loss of generality that $\theta_{11}(n)\!=\!\theta_{22}(n)\!=\!0$, where $\theta_{21}(n)$ and $\theta_{12}(n)$ denote the rapidly varying phase differences. The preceding assumptions are summarized in Table \ref{tCSI}, where an entry ``Y'' (respectively ``N'') means the CSI in the corresponding column is known (respectively unknown) to the terminal in the corresponding row, and an empty entry means whether the corresponding CSI is known to the terminal is inconsequential.


It is assumed that a dedicated link between the two BTSs
allows sharing of their messages, and allocation of powers across streams and code rates may be determined at a BTS or a separate radio control node. Each BTS is subject to its own power constraint.

The proposed cooperation consists of three techniques: 

\subsubsection{Message Sharing and Rate-Splitting}
Each BTS has a message for its assigned mobile, and each BTS may split its
message into two parts, where one part is to be transmitted by itself, and the other
part is shared with and transmitted by the other BTS. This splits the data
stream intended for each mobile. This implies partial message sharing across
the two BTSs, which reduces the burden on the backhaul link relative
to full message sharing.
\subsubsection{Superposition Coding}
Each BTS has its own message intended for its assigned mobile, and
possibly also the shared message from the other BTS.  The two
messages are encoded separately.  Let the message transmitted by BTS
$k$ and intended for mobile $j$ be encoded as $\big(x_{jk}{(1)}, \ldots,
x_{jk}{(N)}\big)$ so that the superposition $x_k{(n)}\!=\!x_{1k}{(n)}\!+\!x_{2k}{(n)}$ is transmitted by BTS $k$.  Then the received signal
by mobile $j\!=\!1,2$ during the $n$-th symbol interval can be rewritten as
\begin{equation}
    y_j{(n)}=\sqrt{g_{j1}}e^{i\theta_{j1}(n)}\big[x_{11}{(n)}+x_{21}{(n)}\big]+\sqrt{g_{j2}}e^{i\theta_{j2}(n)}\big[x_{12}{(n)}+x_{22}{(n)}\big]+z_j{(n)}\label{eq:recsig}
\end{equation}
and the per-BTS power constraints can be stated as
\begin{equation}
\frac{1}{N}\sum_{n=1}^{N}\big[|x_{1k}(n)|^2+|x_{2k}(n)|^2\big]\leq
P_k,~k=1,2.\label{eq:powercon}
\end{equation}

\subsubsection{Interference Cancellation}
Each mobile receives the desired signal, an interference signal, and
noise. Both the desired and interference signals may come from two BTSs.
Here we assume that each mobile treats the
interference signal as noise and does not attempt to decode the
messages intended for the other mobile, since one mobile may not be aware of the
modulation and coding scheme of the signals intended for the other mobile.\footnote{This assumption is consistent with the LTE/LTE-Advanced standards currently under development \cite{SesTou11}.}

Each mobile decodes its two messages, one from each BTS, possibly
using successive decoding, i.e., it can
first decode the message from one BTS and completely cancel the self-interference
when decoding the message from the other BTS. With this scheme the
two-transmitter two-receiver channel can be viewed as two mutually
interfering multiaccess channels (MACs), where each MAC consists
of one mobile receiver and both BTS transmitters.

Throughout this paper, we assume standard Gaussian codebooks
(although those may not be optimal in this scenario), which do not require
phase synchronization among the two cooperating transmitters.

\section{Narrowband Rate Region}
\subsection{The Rate Region Frontier}
Given the channel gains $\boldsymbol{g}\!=\!(g_{11},g_{21},g_{12},g_{22})$ and power constraints $\boldsymbol{P}\!=\!(P_1,P_2)$, the achievable rate region is defined as the convex hull of the following region
\begin{align}
\mathcal{R}(\boldsymbol{g},\boldsymbol{P})\triangleq
\bigg\{&(R_1,R_2)\bigg\arrowvert
\begin{matrix}
R_1=\log\big(1+\frac{g_{11}P_{11}+g_{12}P_{12}}{1+g_{11}P_{21}+g_{12}P_{22}}\big)
\\R_2=\log\big(1+\frac{g_{21}P_{21}+g_{22}P_{22}}{1+g_{21}P_{11}+g_{22}P_{12}}\big)
\end{matrix}\label{eq:Rk}\\
&\text{for some}~P_{11}, P_{12}, P_{21}, P_{22}\!\geq\!0~\text{satisfying}~P_{11}\!+\!P_{21}\!\leq\!
P_1~\text{and}~P_{12}\!+\!P_{22}\!\leq\!P_2\bigg\},\notag
\end{align}
where $P_{jk}$ is the power BTS $k$ allocates to mobile $j$. (It is not difficult to check that $\mathcal{R}(\boldsymbol{g},\boldsymbol{P})$ is indeed a region.) To find the frontier of $\mathcal{R}(\boldsymbol{g},\boldsymbol{P})$, we maximize the
rate of user 2 for a given rate of user 1.  By sweeping the
rate of user 1 over all possible values, we obtain the trade-off
between the two users' rates (without time sharing), which is
referred to as the rate region frontier. From the rate expression
(\ref{eq:Rk}), the optimization problem is written as:
\begin{subequations}
\label{eq:fracLP}
\begin{align}
\underset{\{P_{jk}\}}{\operatorname{\text{maximize}}}~~&R_2=
\log\bigg(1+\frac{g_{21}P_{21}+g_{22}P_{22}}{1+g_{21}P_{11}+g_{22}P_{12}}\bigg)\label{eq:rate}\\
\text{subject~to}~~&\log\bigg(1+\frac{g_{11}P_{11}+g_{12}P_{12}}{1+g_{11}P_{21}+g_{12}P_{22}}\bigg)=R_1\label{eq:con1}\\
&P_{1k}+P_{2k}\leq P_k,~~~P_{jk}\geq 0,~j,k=1,2.\label{eq:BTScon}
\end{align}
\end{subequations}

Since $\log(\cdot)$ is an increasing function, the
rate objective (\ref{eq:rate}) is equivalently the argument of the
logarithm function. The constraint (\ref{eq:con1}) is
linear in the variables. Hence this optimization problem
is a \emph{linear-fractional program} \cite{BoyVan04}.
Specifically, letting $Z\!=\!1\!+\!g_{21}P_{11}\!+\!g_{22}P_{12}$ and
\begin{equation}
P_{jk}=\frac{\widetilde{P}_{jk}}{Z},~j,k=1,2,\label{eq:optpower}
\end{equation}
Problem (\ref{eq:fracLP}) can be rewritten as the following
equivalent linear program:
\begin{subequations}
\label{eq:LP}
\begin{align}
\underset{\{\widetilde{P}_{jk}\},~Z}{\operatorname{\text{maximize}}}~~&g_{21}\widetilde{P}_{21}+g_{22}\widetilde{P}_{22}\\
\text{subject~to}~~&g_{11}\widetilde{P}_{11}+g_{12}\widetilde{P}_{12}-(2^{R_1}-1)(g_{11}\widetilde{P}_{21}+g_{12}\widetilde{P}_{22}+Z)=0\\
&\widetilde{P}_{1k}+\widetilde{P}_{2k}-P_kZ\leq 0,~k=1,2\\
&g_{21}\widetilde{P}_{11}+g_{22}\widetilde{P}_{12}+Z=1,~~~\widetilde{P}_{jk},Z\geq 0,~j,k=1,2.
\end{align}
\end{subequations}
From \cite[Ch. 4]{BoyVan04}, it can be shown that the optimal $Z\!>\!0$.
The optimal power $P_{jk}$ in Problem (\ref{eq:fracLP}) is given by
(\ref{eq:optpower}), where $\widetilde{P}_{jk}$ and $Z$
are obtained from the linear
program in (\ref{eq:LP}), which can be solved
efficiently using standard techniques \cite{BoyVan04}.

The rate region frontier can then be computed by solving a family of
linear-fractional programs corresponding to sweeping over the valid
range of $R_1$. It is easy to show that the valid range of $R_j$
is the interval $\big[0,\log(1\!+\!g_{j1}P_1\!+\!g_{j2}P_2)\big]$.

\begin{figure}
  \center
  \includegraphics[width=0.5\textwidth]{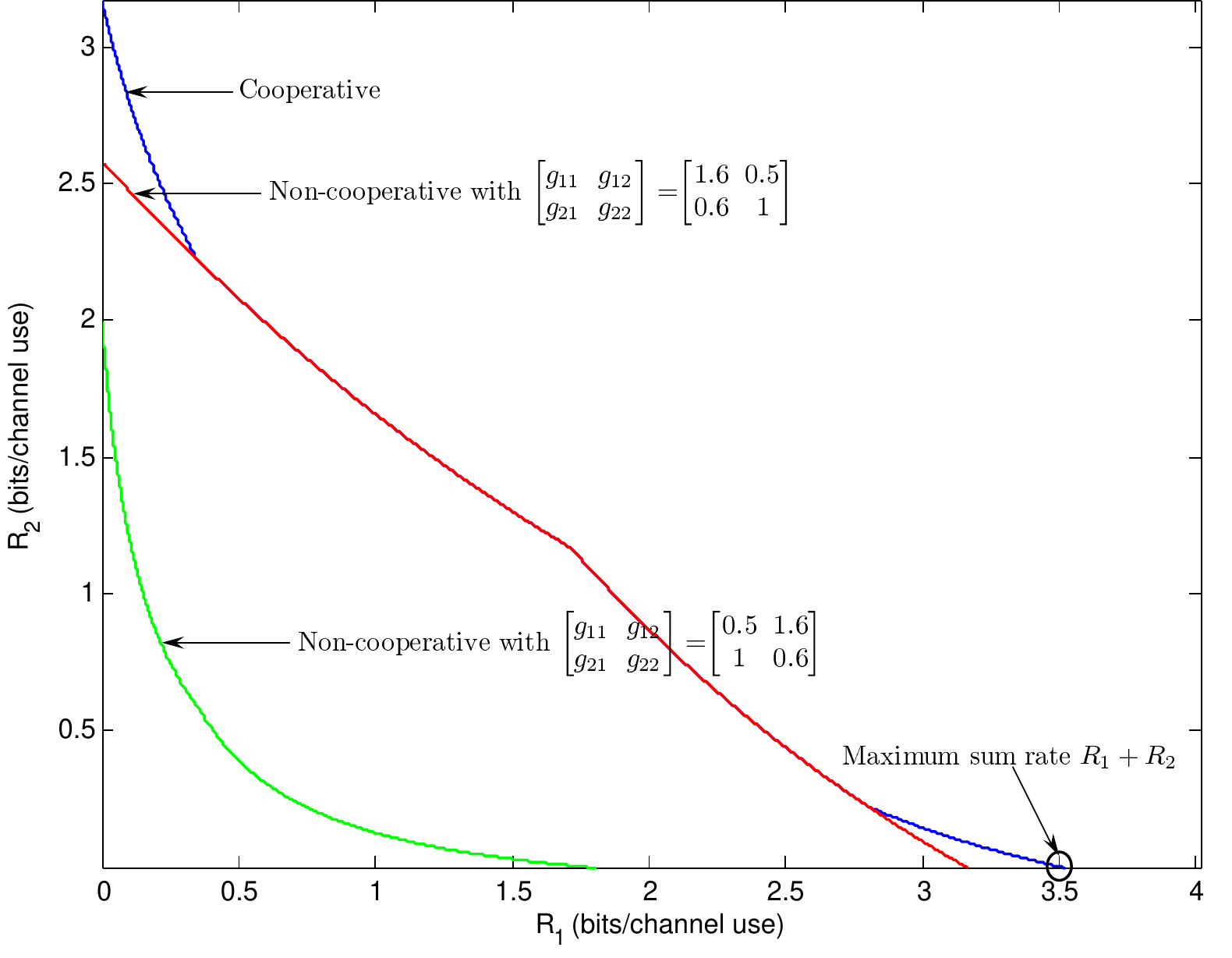}
  \caption{The rate region frontiers achieved by cooperative and non-cooperative schemes with $P_1\!=\!P_2\!=\!5$.}
\label{f:narrow}
\end{figure}
Fig.~\ref{f:narrow} illustrates the rate region frontiers achieved by
the cooperative scheme described in \prettyref{s:narrow}, and by
joint power control (without data sharing) as proposed
in~\cite{ChaSez07Allerton}. Two examples are shown. In the first,
the direct-channel gains are stronger than the cross-channel gains.
The cooperative scheme achieves noticeable gains over the
non-cooperative scheme only when $R_1$ or $R_2$ is near zero.
In the second example, the
direct- and cross-channel gains are swapped, so that the cross-channel
gains are stronger.\footnote{This scenario may not be practical for narrowband systems, because the mobile would automatically switch to the BTS from which it receives stronger signal. However, this serves as a basis for the study of wideband systems in the next section.} The rate region achieved without cooperation
becomes much smaller, while the rate region achieved with cooperation
remains the same. Hence the gain due to cooperation in this scenario
is mainly due to BTS selection.

The following lemma characterizes the optimal power allocation
for points on the rate region frontier. We first observe that at least
one of the constraints in (\ref{eq:BTScon}) must be binding, because
otherwise increasing all $\{P_{jk}\}$ proportionally increases both
$R_1$ and $R_2$.
\newtheorem{lemma}{Lemma}
\begin{lemma}
Every $(R_1,R_2)$ on the rate region
frontier is achieved by a power allocation that satisfies
\begin{equation}
\begin{cases}
P_{1k}+P_{2k}=P_k,~&\text{if}~(2^{R_1}\!-\!1)(2^{R_2}\!-\!1)\leq1\\
P_{1k}P_{2k}=0,~&\text{if}~(2^{R_1}\!-\!1)(2^{R_2}\!-\!1)>1
\end{cases}
\end{equation}
for $k\!=\!1,2$.
\end{lemma}

The proof is given in the appendix.
If $R_1$ or $R_2$ is small enough such that
$(2^{R_1}\!-\!1)(2^{R_2}\!-\!1)\!\leq1\!$, then both BTSs transmit with full
power. There may be points on the rate region frontier that satisfy
$(2^{R_1}\!-\!1)(2^{R_2}\!-\!1)\!>\!1$, which implies $R_1\!>\!0$ and $R_2\!>\!0$.
Each BTS then either transmits only its own message or only the shared message
from the other BTS. These observations provide an easy way to compute
the maximum weighted sum rate, as shown in the next section.

\subsection{Weighted Sum Rate Maximization}
Consider the problem of maximizing the
{\em weighted} sum rate
\begin{equation}
    R(\mu)\triangleq R_1+\mu R_2,\label{eq:sumrate}
\end{equation}
where the user rates are given by (\ref{eq:Rk}) and
$\mu\!\geq\!0$ is the relative priority assigned to the second mobile
and allows for a tradeoff between the overall system throughput and
user fairness. For given $\mu$, the maximum of $R(\mu)$ is always
achieved by some rate pair on the rate region frontier.  For every
$\mu$, (\ref{eq:sumrate}) describes a straight line in the $(R_1,R_2)$
plane, which is an outer bound on the rate region.  The intersection
of the regions below all such outer bounds is exactly the achievable
rate region, which is the convex hull of the region under the rate
region frontier. Any rate pair in this rate region can be achieved
by time sharing two rate pairs on the rate region frontier found
in Section III.A.

\newtheorem{prop}{Proposition}
\begin{prop}
\label{pr:miueq1}
The maximum of $R_1\!+\!R_2$ is achieved at
one of the four corner points listed in Table~\ref{t}.
\end{prop}

The four corner points of the power constraint set shown in
Table~\ref{t} correspond to {\em full cooperation},
meaning that both BTSs cooperatively transmit with full power to
one mobile or each BTS only transmits the shared message from the
other BTS, and {\em non-cooperation}, meaning that each BTS transmits
to its own mobile with full power and without rate-splitting.

\begin{table}
  \caption{Corner points and possible stationary points for weighted sum rate maximization.}
  \begin{center}
  \label{t}
    \begin{tabular}{c|c|cccc}
      &      & $P_{11}$ & $P_{21}$ & $P_{12}$ & $P_{22}$ \\
      \hline
      \multirow{4}{*}{corner points} &
      \multirow{3}{*}{full cooperation} &  $P_1$ &  $0$ &  $P_2$ &  $0$ \\
      \cline{3-6}
      & &  $0$ &  $P_1$ &  $0$ &  $P_2$ \\
      \cline{3-6}
      && $0$ &  $P_1$ &  $P_2$ &  $0$\\
      \cline{2-6}
      &\multirow{1}{*}{non-cooperation} & $P_1$ &  $0$ &  $0$ &  $P_2$ \\
      \hline
      \multirow{4}{*}{stationary points} &
      \multirow{2}{*}{$\mu>1$} &  $P_{11}^*$ &  $0$ &  $0$ &  $P_2$ \\
      \cline{3-6}
      &&  $0$ &  $P_1$ &  $P_{12}^*$ &  $0$ \\
      \cline{2-6}
      &\multirow{2}{*}{$\mu<1$} & $P_1$ &  $0$ &  $0$ &  $P_{22}^*$ \\
      \cline{3-6}
      && $0$ &  $P_{21}^*$ &  $P_2$ &  $0$ \\
      \hline
    \end{tabular}
  \end{center}
\end{table}

\begin{prop} \label{pr:stationary}
If $\mu\!\neq\!1$, then the power allocation,
which maximizes $R_1\!+\mu\!R_2$ satisfies
$P_{1k}P_{2k}\!=\!0$ for $k\!=\!1,2$,
and has the form shown in Table~\ref{t} where $P_{jk}^*\!<\!P_k$.
\end{prop}

\begin{proof}
The proofs of Propositions \ref{pr:miueq1} and \ref{pr:stationary}
consist of examining the stationary points
associated with the two conditions in Lemma 1.
(Note that there must exist a point on the rate frontier
that achieves sum rate $R(\mu)$.) We first show that
if $(2^{R_1}\!-\!1)(2^{R_2}\!-\!1)\leq1$, then $R(\mu)$
achieves its maximum at one of the four corner points listed in Table~\ref{t}.
From (\ref{eq:Rk}) and Lemma 1, we have
\begin{equation}
R(\mu)=\log\bigg[\frac{1+g_{11}P_1+g_{12}P_2}{1+g_{11}(P_1-P_{11})+g_{12}(P_2-P_{12})}\bigg]+\mu\log\bigg[\frac{1+g_{21}P_1+g_{22}P_2}{1+g_{21}P_{11}+g_{22}P_{12}}\bigg].
\end{equation}
It is easy to show that for any fixed $P_{12}$,
$\frac{\partial R(\mu)}{\partial P_{11}}$ is increasing
with $P_{11}$ so that $R(\mu)$ is maximized
at an extreme value for $P_{11}$.
More generally, it is straightforward to show
that $\frac{\partial R(\mu)}{\partial P_{jk}}$
is increasing with $P_{jk}$ for all $j$ and $k$.
Hence $R(\mu)$ is maximized at one of the
extreme points of the power constraint set.

To find stationary points on the rate region frontier
satisfying $(2^{R_1}\!-\!1)(2^{R_2}\!-\!1)\!>\!1$,
Lemma 1 states that we can assume $P_{1k}P_{2k}\!=\!0$.
In general, the rate region frontier may not contain points
satisfying the condition $(2^{R_1}\!-\!1)(2^{R_2}\!-\!1)\!>\!1$.
In that case, the power allocation schemes
satisfying $P_{kk}P_{jk}\!=\!0$ must be suboptimal. However,
without knowing the rate pairs on the frontier, we can assume this
condition is satisfied and characterize the stationary points, which then serve
as candidate points for achieving the maximum weighted sum rate. This
gives two possible frontiers corresponding to the two types of power
allocations in Table~\ref{t}. Namely, one candidate frontier is obtained by fixing
$P_{22}\!=\!P_2$ (or $P_{11}\!=\!P_1$) and sweeping the value of $P_{11}$
(or $P_{22}$) across the interval $[0,P_1]$ (or $[0,P_2]$). The other
candidate frontier is obtained by fixing
$P_{21}\!=\!P_1$ (or $P_{12}\!=\!P_2$) and sweeping the value of $P_{12}$
(or $P_{21}$) over the interval $[0,P_2]$ (or $[0,P_1]$). The actual rate
frontier is then the maximum of the two candidate frontiers.

For the first candidate frontier we have
\begin{equation}
R(\mu)=\log\bigg(1+\frac{g_{11}P_{11}}{1+g_{12}P_2}\bigg)+\mu\log\bigg(1+\frac{g_{22}P_2}{1+g_{21}P_{11}}\bigg).\label{eq:sumpiece1}
\end{equation}
Examining $\frac{dR(\mu)}{dP_{11}}$, the maximizing value $P_{11}^*$ is the solution to the quadratic equation
$a_1 P_{11}^2\!+\!b_1 P_{11}\!+\!c_1\!=\!0$, where $a_1\!=\!g_{11}g_{21}^2\!>\!0$,
$b_1\!=\!2g_{11}g_{21}\!+\!(1-\mu)g_{11}g_{21}g_{22}P_2$, and
$c_1\!=\!g_{11}(1\!+\!g_{22}P_2)\!-\!\mu g_{21}g_{22}P_2(1\!+\!g_{12}P_2)$.
Since $a_1\!>\!0$, the smaller root is the solution.
It is easy to check that if
$\mu$, $g_{jk}$, and $P_2$ satisfy $b_1^2\!-\!4a_1c_1\!\geq\!0$ and $0\!<\!P_{11}^*\!<\!P_1$,
then $P_{11}^*$ achieves the maximum $R(\mu)$.
Similarly, we could fix $P_{11}\!=\!P_1$
and optimize over $P_{22}$. The resulting necessary conditions show that if $\mu\!>\!1$
(or $\mu\!<\!1$), then maximizing over $P_{11}$ (or $P_{22}$) gives a candidate
stationary point, as stated in Proposition~\ref{pr:stationary}.
A similar argument shows that if $\mu\!>\!1$ (or $\mu\!<\!1$),
then maximizing over $P_{12}$ (or $P_{21}$)
gives a second candidate stationary point on the frontier.

If $\mu\!=\!1$, then $b_1\!>\!0$, which implies that
$P_{11}^*\!<\!0$ (if it is real). It can be similarly
verified for the other cases that there are no valid stationary points
in the interior of the power constraint set, which establishes
Proposition \ref{pr:miueq1}.
\end{proof}

The preceding propositions state that the maximum weighted sum rate
can be efficiently determined by searching over the small number of
candidate power allocations shown in Table~\ref{t}.
This will be used as the basis for optimizing wideband power allocations
discussed in the next section.

\section{Frequency-Selective Channels}
\subsection{Problem Formulation}
We now consider a wideband system with frequency-selective channels.
A wideband channel is modeled as a set of $L$ discrete
(parallel) channels. Each sub-channel is modeled similarly as (\ref{eq:recsig1}), with the same CSI known at the terminals.\footnote{The phase difference is in fact identical over different sub-channels.} Instead of one power constraint for each sub-channel, the $L$ sub-channels are subject to a total power constraint at each BTS. The problem is to maximize the weighted sum
across users of the rates summed across sub-channels:
\begin{subequations}
\label{eq:problem1}
\begin{align}
\underset{\{P_{jk}(l),P_{k}(l)\}}{\operatorname{\text{maximize}}}&~~\sum_{l=1}^{L}\big[R_1 (l) + \mu R_2 (l)\big]\\
\text{subject to}&~~P_{1k}(l)+P_{2k}(l)\leq P_{k}(l),~\forall k,l\\
        &~~\sum_{l=1}^{L}P_k(l)\leq P_{\text{tot},k},~~~P_{jk}(l)\geq0,~~~P_k(l)\geq0,~~~\forall j,k,l,\label{eq:problem1Con}
\end{align}
\end{subequations}
where $l$ denotes the sub-channel index, $P_k(l)$ denotes the power
allocated to sub-channel $l$, and $P_{\text{tot},k}$ denotes the total
power constraint at BTS $k$. Rates $R_1 (l)$ and $R_2 (l)$
are given by (\ref{eq:Rk}), where the channel gains
$g_{jk}$ and powers depend on $l$.

This can be viewed as a two-level optimization problem. At the
lower level the weighted sum rate is maximized for each
sub-channel given its allocated power. The upper level then
optimizes the power allocation across sub-channels
subject to the total power constraints.
Based on the discussion in the last section,
the maximum rate for each sub-channel is
achieved by one of the cases in Table~\ref{t}. In general, solving the
two-level problem requires iterating between the lower- and upper-levels.

The three cooperative power assignments in Table~\ref{t} give the
following weighted sum rates for sub-channel $l$:
\begin{align}
    R_1^{(\text{c})}(l)\triangleq&\log\big[1+g_{11}(l)P_1(l)+g_{12}(l)P_2(l)\big]\label{eq:rc1}\\
    R_2^{(\text{c})} (l)\triangleq&\mu \log\big[1+g_{21}(l)P_1(l)+g_{22}(l)P_2(l)\big]\label{eq:rc2}\\
    R_3^{(\text{c})}(l)\triangleq&\log\bigg[1+\frac{g_{12}(l)P_2(l)}{1+g_{11}(l)P_1(l)}\bigg]+\mu\log\bigg[1+\frac{g_{21}(l)P_1(l)}{1+g_{22}(l)P_2(l)}\bigg],\label{eq:rc3}
\end{align}
whereas the non-cooperative assignment gives
\begin{equation}
\label{eq:rnc}
R^{\text{(nc)}}(l)\triangleq\log\bigg[1+\frac{g_{11}(l)P_1(l)}{1+g_{12}(l)P_2(l)}\bigg]+\mu\log\bigg[1+\frac{g_{22}(l)P_2(l)}{1+g_{11}(l)P_1(l)}\bigg].
\end{equation}
The power control problem is then
\begin{equation}
\label{eq:problem2}
        \underset{\{P_{k}(l)\}}{\operatorname{\text{maximize}}}~~\sum_{l=1}^{L}
        \max\bigg\{R_1^{(\text{c})}(l),R_2^{(\text{c})}(l),R_3^{(\text{c})}(l),R^{(\text{nc})}(l)\bigg\}
\end{equation}
subject to (\ref{eq:problem1Con}).

Note that the rate objective includes only the corner points and
does not explicitly include the interior points (stationary points).
However, the interior points are implicitly included in the rate
objective due to the power optimization at the upper level. Namely,
if the weighted sum rate for sub-channel $l$ is maximized at an
interior point, e.g., corresponding to $P_{11}^* (l)\!<\!\overline{P}_1
(l)$, $P_{22}(l)\!=\!\overline{P}_2 (l)$, $P_{12}(l)\!=\!P_{21}(l)\!=\!0$, where
$\overline{P}_k (l)$ is the $l$-th sub-channel power constraint at BTS
$k$, then the rate can be increased by decreasing $\overline{P}_1 (l)$
via the upper-level optimization.

\subsection{Continuous Power Allocation}
Problem (\ref{eq:problem2}) is non-convex in general because
of the non-convexity of the objective function. However, letting the
number of sub-carriers within a given band $F$
go to infinity, we can assume that $g_{jk}(l)$ converges to a continuous
function of frequency $f\!\in\!F$, and the corresponding continuous
optimization problem can be efficiently solved numerically. The continuous optimization problem can be
formulated as
\begin{subequations}
\label{eq:problem3}
\begin{align}
\underset{\{P_{k}(f)\}}{\operatorname{\text{maximize}}}&~~\int_{F}\max\bigg\{R_1^{(\text{c})}(f),R_2^{(\text{c})}(f),R_3^{(\text{c})}(f),R^{(\text{nc})}(f)\bigg\} d f\\
\text{subject to}&~~\int_{F}P_k(f)df\leq
P_{\text{tot},k},~~~P_k(f)\geq0,~\forall k,f,\label{eq:problem3Con}
\end{align}
\end{subequations}
where the index $l$ is replaced by the continuous variable $f$.

\newtheorem{definition}{Definition}
{\em Definition \cite{YuLui06TC}:}
Consider the general optimization problem:
\begin{subequations}
\begin{align}
\text{maximize}&~~\sum_{n=1}^{N}f_n(\boldsymbol{x}_n)\\
\text{subject to}&~~\sum_{n=1}^{N}\boldsymbol{h}_n(\boldsymbol{x}_n)\leq\boldsymbol{P},
\end{align}\label{eq:timesharing}
\end{subequations}
where $\boldsymbol{x}_n\in \mathbb{R}^K$ are the optimization variables, each function $f_n(\cdot)\!:~\mathbb{R}^K\!\rightarrow\!\mathbb{R}$ is not necessarily concave, and each function $\boldsymbol{h}_n(\cdot)\!\!:\mathbb{R}^K\!\rightarrow\!\mathbb{R}^L$ is not necessarily convex. Power constraints are denoted by an $L$-vector $\boldsymbol{P}$. Here, ``$\leq$'' is used to denote a component-wise inequality.
An optimization problem of the form (\ref{eq:timesharing}) is said to satisfy the {\em time-sharing condition} if for any $\boldsymbol{P}_{\boldsymbol{x}}$, $\boldsymbol{P}_{\boldsymbol{y}}$ with corresponding optimal solutions $\boldsymbol{x}_n^*$ and $\boldsymbol{y}_n^*$, respectively, and for any $0\leq v\leq1$, there exists a feasible solution $\boldsymbol{z}_n$, such that $\sum_{n=1}^N\boldsymbol{h}_n(\boldsymbol{z}_n)\leq \nu\boldsymbol{P}_{\boldsymbol{x}}+(1-\nu)\boldsymbol{P}_{\boldsymbol{y}}$, and $\sum_{n=1}^Nf_n(\boldsymbol{z}_n)\geq \nu\sum_{n=1}^Nf_n(\boldsymbol{x}_n^*)+(1-\nu)\sum_{n=1}^Nf_n(\boldsymbol{y}_n^*)$.

The time-sharing condition essentially states that the optimal
value of the objective in (\ref{eq:timesharing}) is concave in
$\boldsymbol{P}$. It is shown in \cite{YuLui06TC} that if the time-sharing
condition is satisfied, then the optimization problem has zero
duality gap, i.e., solving the dual problem gives the same optimal value as solving the primal problem even if it is not convex.

Since $g_{jk}(f)$ is continuous in $f$, $R_{i}^{(\text{c})}(f),~i\!=\!1,2,3$ and
$R^{(\text{nc})}(f)$ are continuous in $f$, and the integrand of the
objective function is continuous in $f$. Therefore, we can apply the techniques used in the proof of Theorem 2 in \cite{YuLui06TC} to
show that this optimization problem satisfies the time-sharing condition.\footnote{The objective function of the problem considered in Theorem 2 \cite{YuLui06TC} is different from that considered here, but as long as the integrand in the objective function is continuous, the proof of Theorem 2 still applies.} Then the optimization problem (\ref{eq:problem3}) has zero
duality gap and for this problem solving its
dual is more efficient.


The solution to Problem (\ref{eq:problem3}) approximates
the solution to (\ref{eq:problem2}) and it becomes more accurate as
$L\!\rightarrow\!\infty$. In practical systems with a large number of
sub-carriers, the channel gains between consecutive sub-carriers are
typically highly correlated. Hence we expect that the time-sharing
condition is approximately satisfied for (\ref{eq:problem2}), so that the
solution to the dual problem will be nearly-optimal.

The Lagrangian function associated with Problem (\ref{eq:problem2}) is
\begin{align}
L(\boldsymbol{P}_1,&\boldsymbol{P}_2,\boldsymbol{\lambda})=\sum_{l=1}^{L}\max\bigg\{R_1^{(\text{c})}(l),R_2^{(\text{c})}(l),R_3^{(\text{c})}(l),R^{(\text{nc})}(l)\bigg\}\notag\\
&-\lambda_1\bigg(\sum_{l=1}^{L}P_1(l)-P_{\text{tot},1}\bigg)-\lambda_2\bigg(\sum_{l=1}^{L}P_2(l)-P_{\text{tot},2}\bigg),
\end{align}
where $\lambda_k\!\geq\!0$ is the
Lagrange multiplier associated with the power constraint for BTS
$k$ and the bold-font denotes the vector version of the corresponding
variables. The dual optimization problem associated
with Problem (\ref{eq:problem2}) is
\begin{equation*}
\underset{\lambda_1,\lambda_2}{\operatorname{\text{minimize}}}~~g(\boldsymbol{\lambda})~~\text{subject
to}~\lambda_1\geq0,~\lambda_2\geq0,
\end{equation*}
where
\begin{equation}
g(\boldsymbol{\lambda})=\displaystyle{\max_{\boldsymbol{P}_1,\boldsymbol{P}_2}}L(\boldsymbol{P}_1,\boldsymbol{P}_2,\boldsymbol{\lambda})\label{eq:dual}
\end{equation}
is the dual objective function.

Compared with numerically solving the primal problem (\ref{eq:problem2})
directly, two properties of the dual
problem lead to a reduction in computational complexity. One is
that for any fixed $\boldsymbol{\lambda}$, the solution to
(\ref{eq:dual}) can be computed per-sub-carrier since
(\ref{eq:dual}) can be decomposed into $L$ parallel unconstrained
optimization problems. Note that for each sub-carrier an exhaustive
search for the optimal $P_k(l)$ must be carried out. The other property
is the dual problem is convex in the variables
$\boldsymbol{\lambda}$, which guarantees the convergence of
numerical methods. In the numerical results that follow we solve the
dual optimization problem efficiently via a nested bisection search over
$\lambda_1$ and $\lambda_2$ \cite{YuLui06TC,CenYu06TC}. (An outer loop updates
$\lambda_1$ and an inner loop updates $\lambda_2$.) If the required accuracy of
each $\lambda$ is given by $\epsilon_\lambda$, then the overall computational
complexity of the bisection search is $\mathcal{O}(L\log(1/\epsilon_\lambda)^2)$,
which is linear in $L$ \cite{CenYu06TC}.

\subsection{High-SNR Approximation with $\mu\!=\!1$}
At high signal-to-noise ratios (SNRs), we have $R_i^{(\text{c})}\!>\!R^{(\text{nc})}$ and
$R_i^{(\text{c})}\!>\!R_3^{(\text{c})}$ for $i\!=\!1,2$, since
$R^{(\text{nc})}$ and $R_3^{(\text{c})}$ are interference-limited.
Hence we can simplify Problem \eqref{eq:problem3} by maximizing over only $R_i^{(\text{c})}$,
$i\!=\!1,2$, in the integrand of the rate objective. The corresponding
suboptimal power allocation problem can be written as
\begin{equation}
\label{eq:problem4}
\underset{\{P_k(f)\}}{\operatorname{\text{maximize}}}~~\int_{F}\max\bigg\{R_1^{(\text{c})}(f),R_2^{(\text{c})}(f)\bigg\}df
\end{equation}
subject to (\ref{eq:problem3Con}). This is still a two-level optimization problem. The lower level selects
the better MAC channel from the two options for each
sub-carrier, and the upper level distributes the power to each
sub-carrier subject to the total power constraint.

Although $R_1^{(\text{c})}(f)$ and $R_2^{(\text{c})}(f)$ are concave
functions of $P_1(f)$ and $P_2(f)$,
$\max\big\{R_1^{(\text{c})}(f),R_2^{(\text{c})}(f)\big\}$ is in
general not concave in $P_1(f)$ and $P_2(f)$. Problem (\ref{eq:problem4})
can be efficiently solved as described in the last subsection
(via its dual problem); however, it turns out that under some mild
conditions it can be transformed into a convex program so that
standard efficient numerical techniques can also be applied.

Specifically, we solve this problem by finding a convex function that upper bounds $\max\big\{R_1^{(\text{c})}(f),$\\$R_2^{(\text{c})}(f)\big\}$, and
optimizing this upper bound over the power allocations. It can be
shown that substituting the optimized power
allocation for the upper bound into the original objective in
(\ref{eq:problem4}) gives the same value as the optimized upper bound.

Letting
\begin{equation}
R_{\text{ub}}(f)\triangleq\log\big[1+\max\{g_{11}(f),g_{21}(f)\}P_1(f)+\max\{g_{12}(f),g_{22}(f)\}P_2(f)\big],\label{eq:ubrate}
\end{equation}
we observe that $R_{\text{ub}}(f)$ serves as a pointwise upper
bound for any full-cooperation scheme, i.e.,
\begin{equation}
\max\bigg\{R_1^{(\text{c})}(f),R_2^{(\text{c})}(f)\bigg\}\leq
R_{\text{ub}}(f),~\forall f.\label{eq:upperbound}
\end{equation}
We now consider the optimization problem
\begin{equation}
\label{eq:problem5}
\underset{\{P_k(f)\}}{\operatorname{\text{maximize}}}~~\int_FR_{\text{ub}}(f)df
\end{equation}
subject to (\ref{eq:problem3Con}).
Since $R_{\text{ub}}(f)$ is concave with respect to $P_1(f)$ and
$P_2(f)$, Problem (\ref{eq:problem5}) is a
convex optimization problem.\footnote{Similar problems have been considered in
\cite{TseHan98IT,KnoHum95ICC}. Here the difference is that in
(\ref{eq:ubrate}) for each branch of the MAC channel and each $f$, there are two optional
channel gains to be selected.} Therefore the necessary
Karush-Kuhn-Tucker (KKT) conditions for optimality are also
sufficient. Letting $y_k(f)\!=\!1/\max\{g_{1k}(f),g_{2k}(f)\}$, the KKT conditions
for the optimal power allocation,
$P_1^*(f)$ and $P_2^*(f)$, can be stated as
\begin{align*}
\text{a}.~&\text{If}~\lambda_1y_1(f)<\lambda_2y_2(f),~\text{then}\\
&P_1^*(f)=\bigg(\frac{1}{\lambda_1}-y_1(f)\bigg)^{+},~P_2^*(f)=0\\
\text{b}.~&\text{If}~\lambda_1y_1(f)>\lambda_2y_2(f),~\text{then}\\
&P_1^*(f)=0,~P_2^*(f)=\bigg(\frac{1}{\lambda_2}-y_2(f)\bigg)^{+}\\
\text{c}.~&\text{If}~\lambda_1y_1(f)=\lambda_2y_2(f),~\text{then}\\
&P_1^*(f)+\frac{y_1(f)}{y_2(f)}P_2^*(f)=\bigg(\frac{1}{\lambda_1}-y_1(f)\bigg)^{+},~P_1^*(f),P_2^*(f)\geq0,
\end{align*}
where $(x)^{+}\!\triangleq\!\max\{x,0\}$, $\lambda_1$ and $\lambda_2$
are non-negative and can be determined by substituting the optimal
power allocation into the power constraints.

For a given set of channel fading gains $\boldsymbol{g}_f$, the
optimal power allocation is not unique only for the preceding Case c.
Assuming the joint distribution of the channel gains is continuous
(e.g., Rayleigh fading), this happens with probability zero.
Therefore with probability one the optimal power allocation is
unique and is determined by the first two conditions. The optimal
power control scheme implies that only one BTS is assigned to transmit
at any given $f$, i.e., the BTS with relatively stronger direct- or
cross-channel gains.
The two-level water-filling structure of the power allocation indicates that orthogonal transmission is optimal and the gain of the cooperative scheme as considered here in the wideband channel comes from cell selection for each sub-channel, since for each sub-channel at most one link among the four direct- and cross-links between BTSs and mobiles is activated.

It is straightforward to show that substituting
the optimized $P_1^*(f)$ and
$P_2^*(f)$ in the objective in (\ref{eq:problem4})
gives the same result as substituting those functions into the
corresponding upper bound $R_{\text{ub}} (f)$.
This is because the solution states that only one BTS transmits at any given $f$.
Since $P_1^*(f)$ and $P_2^*(f)$
maximize the upper bound, they must also maximize
the original objective.

We observe that
$\big[P_1^*(f),P_2^*(f)\big]$ is a KKT point for Problem (\ref{eq:problem4}).
This is because when substituting
$\big[P_1^*(f),P_2^*(f)\big]$ into the four rates listed in
(\ref{eq:rc1})--(\ref{eq:rnc}), the maximum
value among the four rates is equal to the maximum value of the
two cooperative rates in (\ref{eq:rc1}) and (\ref{eq:rc2}). Since
$\big[P_1^*(f),P_2^*(f)\big]$ is
a KKT point for Problem (\ref{eq:problem4}), it must also be a KKT
point for Problem (\ref{eq:problem3}). This implies that
$\big[P_1^*(f),P_2^*(f)\big]$ is a locally optimal power
allocation scheme for (\ref{eq:problem3}) although it may not be globally
optimal.

We emphasize that the equivalence between Problems (\ref{eq:problem4})
and (\ref{eq:problem5}) relies on the continuity
of the integrand in the objective function. With a finite number of
sub-carriers, the solutions to the two problems may not be the same.
For example, with only one sub-carrier, by
inspection the optimized objective in (\ref{eq:problem5}) is $\log\big(1\!+\!\max\{g_{11},g_{21}\}P_{\text{tot},1}\!+\!\max\{g_{12},g_{22}\}P_{\text{tot},2}\big)$,
which cannot be achieved by (\ref{eq:problem4}) in general.

\begin{figure}
  \center
  \subfigure[]{
   \includegraphics[width=0.5\textwidth]{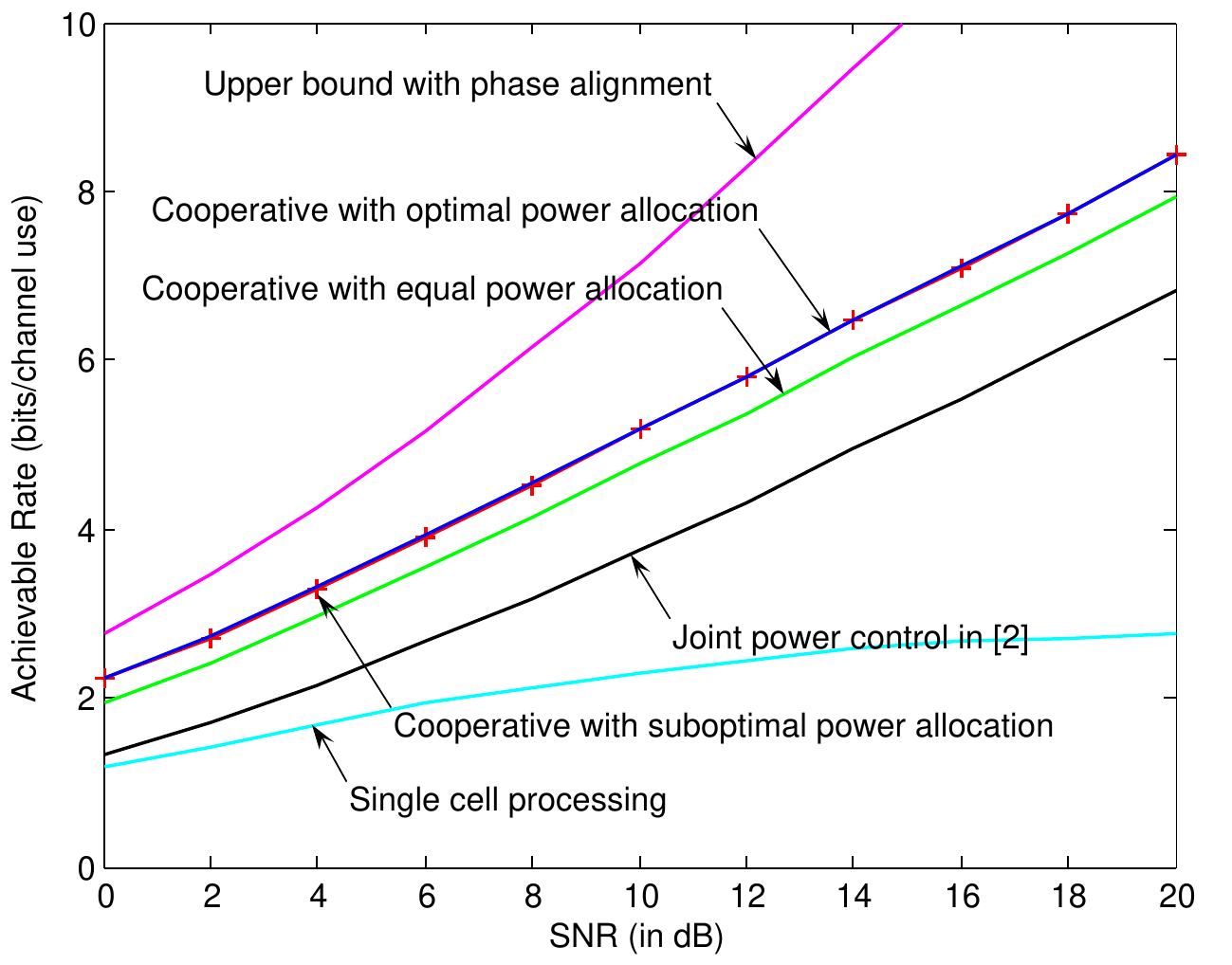}
   \label{f:wide}
   }
  \subfigure[]{
   \includegraphics[width=0.5\textwidth]{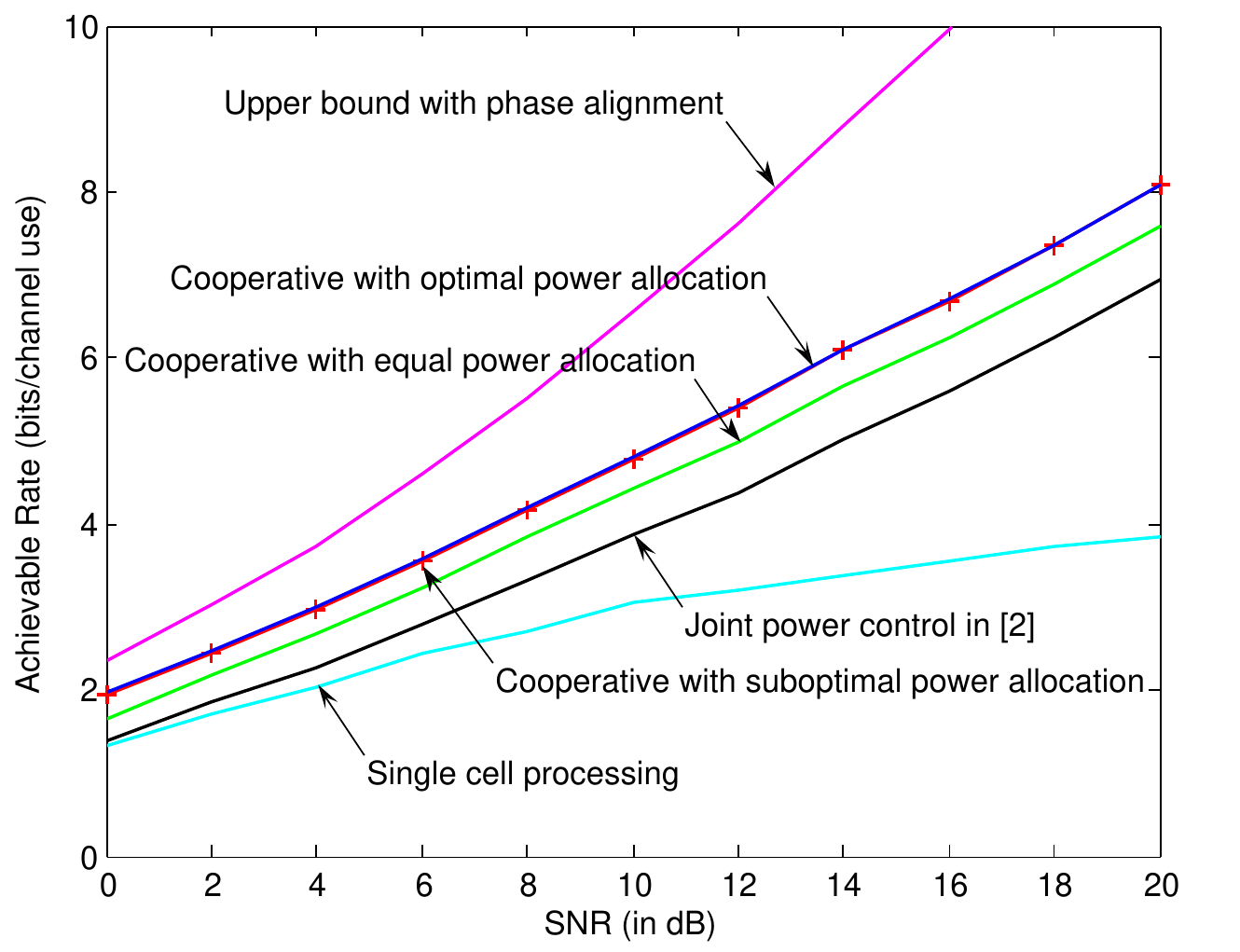}
   \label{f:wide_3dB}
   }
  \caption{The achievable sum rate with: (a) equal direct- and cross-channel gains, i.e., $\mathbb{E}[g_{jk}]\!=\!1,~\forall j,k$; (b) direct-channels are 3 dB stronger than cross-channels, i.e., $\mathbb{E}[g_{jj}]\!=\!1,~j=1,2$ and $\mathbb{E}[g_{jk}]\!=\!0.5,~j,k\!=\!1,2,~j\neq k$.}
  \label{f:s_antenna}
\end{figure}
In Fig.~\ref{f:s_antenna} we compare the maximum sum rates of
the cooperative and non-cooperative schemes with wideband channels
and $\mu\!=\!1$. For this example there are $L\!=\!128$
sub-carriers. 
The channels on each sub-carrier across the four direct- and cross-links undergo independent Rayleigh fading, and for each link the channels across sub-carriers are assumed to have correlation coefficients of $0.95$. The figure compares achievable rates
for the following scenarios: 1) optimized power assignments across
sub-carriers and both BTSs according to (\ref{eq:problem3}); 2) optimized power
assignments across sub-carriers and both BTSs according to (\ref{eq:problem4}); 3) cooperation
between BTSs with equal power assignments across sub-carriers; 4) both BTSs carry out
joint power control but do not assist each other
with data transmissions \cite{GjeGes08TWC}. The achievable downlink sum rate of
perfect BTS cooperation with phase alignment is included for comparison.
Fig.~\ref{f:wide} shows results with equal-variance direct- and
cross-channel gains, which corresponds to the scenario where the
two mobiles are both located close to the cell boundary, and
Fig.~\ref{f:wide_3dB} shows results for the case where the cross-channel
gains are 3 dB weaker than the direct-channel gains.

The results in Fig.~\ref{f:wide} show that the cooperative scheme
considered offers approximately 5 dB gain relative to non-cooperative
joint power control scheme in \cite{GjeGes08TWC}. Also, cooperation
with wideband power allocation offers about one dB gain with respect
to equal power allocation and there is negligible difference
between the suboptimal and optimal wideband power allocations.

The cooperative scheme considered can only
provide diversity gain, and therefore achieves only one degree of
freedom (asymptotic slope of rate curves in Fig.~\ref{f:s_antenna}).
In contrast, if the BTSs can cooperate with phase alignment, then
two degrees of freedom can be achieved
as illustrated in the figure. The performance
improvement due to cooperation diminishes if the average cross-channel
gains become weaker than the direct gains, as illustrated in Fig.~\ref{f:wide_3dB}.


\section{Multiple Transmit Antennas}
We now consider the case where each BTS has $N_t\!\geq\!2$ transmit
antennas and each mobile has a single receive antenna. Assuming a
narrowband system with block fading, the baseband signal received by
mobile $j\!=\!1,2$ during the $n$-th symbol interval is
\begin{equation}
    y_j{(n)}=\big(\boldsymbol{h}_{j1}e^{i\theta_{j1}(n)}\big)^{\dagger}\boldsymbol{x}_{1}{(n)}+\big(\boldsymbol{h}_{j2}e^{i\theta_{j2}(n)}\big)^{\dagger}\boldsymbol{x}_{2}{(n)}+z_j{(n)}\label{eq:recsig2},
\end{equation}
where $\boldsymbol{x}_k{(n)}$, for $k\!=\!1,2$, denotes the transmit signal vector of dimension $N_t\!\times\!1$ from BTS $k$, $\boldsymbol{h}_{jk}$ denotes the complex channel vector consisting of the fading coefficients from the transmit antennas at BTS
$k$ to the receive antenna at mobile $j$, and $\theta_{jk}(n)$ denotes the corresponding rapidly changing phase caused by the drifting frequency offset of the local oscillator. The same drift is experienced by all antennas. Similarly, it is assumed that
the complex vectors $(\boldsymbol{h}_{j1},\boldsymbol{h}_{j2})$ are known
to both BTSs and mobile $j$ and remain constant within one coding block, while the phases $\big(\theta_{j1}(n),\theta_{j2}(n)\big)$ are known to mobile $j$ and unknown to the other BTS. Without loss of generality, we can assume $\theta_{11}(n)\!=\!\theta_{22}(n)\!=\!0$ for all $n$.

\subsection{Cooperative Beamforming}
In analogy with the single transmit antenna case, in the cooperative
scheme each BTS splits its message into two parts,
where one part is transmitted by itself and the other part is shared
with and transmitted by the other BTS. Each BTS transmits a
superposition of two codewords intended for the two mobiles and each
codeword consists of scalar coding followed by beamforming, which
can be written as
\begin{equation}
\boldsymbol{x}_k(n)=\boldsymbol{v}_{1k}x_{1k}{(n)}+\boldsymbol{v}_{2k}x_{2k}{(n)},~k=1,2,
\end{equation}
where $\boldsymbol{v}_{jk}$ is a $N_t\!\times\!1$ vector with
$\|\boldsymbol{v}_{jk}\|\!=\!1$, and denotes the normalized beamforming
vector for the scalar symbol $x_{jk}$. The
per-BTS power constraints are again given by (\ref{eq:powercon}).
Each mobile decodes its two messages successively,
treating signals for the other mobile as background noise. This
scheme achieves the rate pair $(R_1^{(\text{bf})},R_2^{(\text{bf})})$ where
\begin{equation}
R_{k}^{(\text{bf})}=\log\bigg(1+\frac{|\boldsymbol{h}_{k1}^{\dagger}\boldsymbol{v}_{k1}|^2P_{k1}+|\boldsymbol{h}_{k2}^{\dagger}\boldsymbol{v}_{k2}|^2P_{k2}}
{1+|\boldsymbol{h}_{k1}^{\dagger}\boldsymbol{v}_{j1}|^2P_{j1}+|\boldsymbol{h}_{k2}^{\dagger}\boldsymbol{v}_{j2}|^2P_{j2}}\bigg),~j\neq k.
\end{equation}

Define
$\bar{\boldsymbol{h}}_{jk}\!\triangleq\!\frac{\boldsymbol{h}_{jk}}{\|\boldsymbol{h}_{jk}\|}$
and $g_{jk}\!\triangleq\!\|\boldsymbol{h}_{jk}\|^2$. The optimal
beamforming vector $\boldsymbol{v}_{jk}$ must lie in the space
spanned by $\bar{\boldsymbol{h}}_{1k}$ and
$\bar{\boldsymbol{h}}_{2k}$ \cite{JorLar08TSP}, because any power spent on the null space of $\bar{\boldsymbol{h}}_{1k}$ and
$\bar{\boldsymbol{h}}_{2k}$ will not be received by any mobile, and it does not have any impact on the signal-to-interference-plus-noise ratio (SINR) at each mobile. Fig.~\ref{f:illusbf}
illustrates the beamforming vector $\boldsymbol{v}_{11}$ geometrically in the
plane spanned by $\bar{\boldsymbol{h}}_{11}$ and $\bar{\boldsymbol{h}}_{21}$, where $\alpha_1\!\triangleq\!\cos^{-1}\!\big|\bar{\boldsymbol{h}}_{11}^\dagger\bar{\boldsymbol{h}}_{21}\big|$
denotes the angle between $\bar{\boldsymbol{h}}_{11}$ and $\bar{\boldsymbol{h}}_{21}$,
and $\beta_{11}$ denotes the angle between $\boldsymbol{v}_{11}$ and $\bar{\boldsymbol{h}}_{21}$.
Then $\boldsymbol{v}_{11}$ can be parameterized as
\begin{align}
\boldsymbol{v}_{11}=\frac{\cos{\beta_{11}}}{\cos{\alpha_1}}\boldsymbol{\Pi}_{\bar{\boldsymbol{h}}_{21}}\bar{\boldsymbol{h}}_{11}
+\frac{\sin{\beta_{11}}}{\sin{\alpha_1}}\boldsymbol{\Pi}^{\bot}_{\bar{\boldsymbol{h}}_{21}}\bar{\boldsymbol{h}}_{11},\label{eq:beamformer}
\end{align}
where $\boldsymbol{\Pi}_{\bar{\boldsymbol{h}}_{21}}\!\triangleq\!\bar{\boldsymbol{h}}_{21}\bar{\boldsymbol{h}}_{21}^{\dagger}$
and $\boldsymbol{\Pi}^{\bot}_{\bar{\boldsymbol{h}}_{21}}\!\triangleq\!
\boldsymbol{I}\!-\!\boldsymbol{\Pi}_{\bar{\boldsymbol{h}}_{21}}$ denote the
projection and orthogonal projection onto the column space of
$\bar{\boldsymbol{h}}_{21}$, respectively, and the terms $\cos{\alpha_1}$
and $\sin{\alpha_1}$ in the corresponding denominators are used to normalize the two orthogonal
vectors $\boldsymbol{\Pi}_{\bar{\boldsymbol{h}}_{21}}\bar{\boldsymbol{h}}_{11}$ and $\boldsymbol{\Pi}^{\bot}_{\bar{\boldsymbol{h}}_{21}}\bar{\boldsymbol{h}}_{11}$.
\begin{figure}
  \center
  \includegraphics[width=0.28\textwidth]{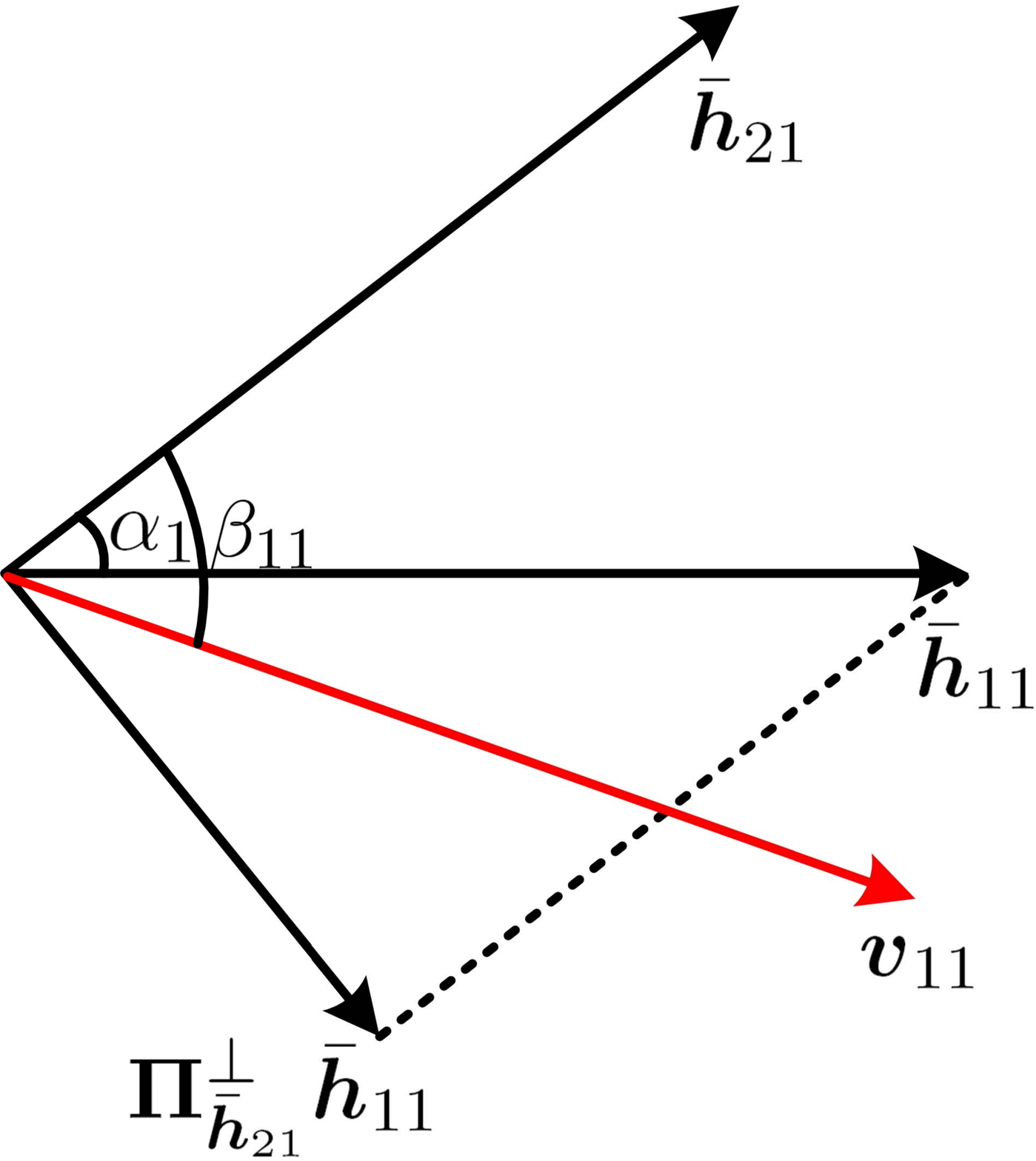}
  \caption{Illustration of the beamforming vector $\boldsymbol{v}_{11}$.}
  \label{f:illusbf}
\end{figure}

From (\ref{eq:beamformer}), we can write $
|\boldsymbol{h}_{11}^{\dagger}\boldsymbol{v}_{11}|^2\!=\!g_{11}\!\cos^2(\beta_{11}\!\!-\!\!\alpha_1)$
and $|\boldsymbol{h}_{21}^{\dagger}\boldsymbol{v}_{11}|^2\!=\!g_{21}\!\cos^2\!\beta_{11}$.
Similarly, the other three beamforming vectors can be parameterized
by introducing
$\alpha_2\!\triangleq\!\cos^{-1}\!\big|\bar{\boldsymbol{h}}_{12}^\dagger\bar{\boldsymbol{h}}_{22}\big|$
and $\beta_{jk}\in[0,\frac{\pi}{2}]$, $j,k =1,2$.
The achievable rate pair $(R_1^{(\text{bf})},R_2^{(\text{bf})})$ can then be re-stated as
\begin{equation}
R_{k}^{(\text{bf})}\!=\!\log\bigg[1\!+\!\frac{g_{k1}\!\cos^2(\beta_{k1}\!\!-\!\!\alpha_1)P_{k1}\!+\!g_{k2}\!\cos^2(\beta_{k2}\!\!-\!\!\alpha_2)P_{k2}}{1\!+\!g_{k1}\!\cos^2\!\beta_{j1}P_{j1}\!+\!g_{k2}\!\cos^2\!\beta_{j2}P_{j2}}\bigg],~j\neq k.\label{eq:beamRk}\\
\end{equation}
Optimizing the beamforming vectors $\{\boldsymbol{v}_{jk}\}$ is
now equivalent to optimizing the corresponding angles
$\{\beta_{jk}\},~j,k\!=\!1,2$. If $\beta_{jk}\!=\!\alpha_k$ then BTS $k$
transmits to mobile $j$ with a maximum-ratio beamformer and
if $\beta_{jk}\!=\!\frac{\pi}{2}$ then BTS $k$ transmits to mobile
$j$ with a zero-forcing beamformer. In general, the optimal beamforming
vectors must strike a balance between these two extremes. At high SNRs
the solution should be close to zero-forcing, and at low SNRs the
solution should be close to maximum-ratio combining. Note that with
$N_t\!\geq2\!$ antennas, the interference term in the denominator of
($\ref{eq:beamRk}$) can be nulled
out by choosing $\beta_{jk}\!=\!\frac{\pi}{2},~j,k\!=\!1,2$,
therefore two degrees of freedom can be achieved.

\subsection{The Achievable Rate Region}
The achievable rate region can be obtained by maximizing the
weighted sum rate $R^{(\text{bf})}(\mu)\triangleq R_1^{(\text{bf})}+\mu R_2^{(\text{bf})},~\mu\geq0$
over the beamforming vectors and the power allocated to each
message for each $\mu$ and then sweeping $\mu$. To achieve a rate
pair on the boundary of the rate region, the beams and powers must
be jointly optimized. The following proposition states that both BTSs should always
transmit with full power. 
\begin{prop}
\label{propbeam} For every rate pair
$(R_1^{(\text{bf})},R_2^{(\text{bf})})$ on the boundary of the rate
region, the corresponding power allocation satisfies
$P_{1k}\!+\!P_{2k}\!=\!P_k,~\forall k$.
\end{prop}
\begin{proof}
This follows from the observation that each beam contains a component,
which is orthogonal to the cross-channel. Hence increasing power
along that component increases the desired power without
increasing interference.
Specifically, let $\{\beta_{jk},P_{jk}\},j,k\!\!=\!\!1,2$
be the optimal parameters for a rate pair on the boundary of the
rate region. From (\ref{eq:beamRk}), the
useful signal power from BTS $k$ to mobile $k$ is
$g_{kk}\!\cos^2(\beta_{kk}\!\!-\!\!\alpha_k)P_{kk}$ and the
corresponding interference power is $g_{jk}\!\cos^2\!\beta_{kk}P_{kk}$,
$j\!\neq\!k$. If $\beta_{kk}\!=\!\frac{\pi}{2}$, then
increasing $P_{kk}$ will increase $R_k^{(\text{bf})}$
without changing $R_j^{(\text{bf})}$.
If $\beta_{kk}\!\neq\!\frac{\pi}{2}$, then with fixed interference,
i.e., $g_{jk}\!\cos^2\!\beta_{kk}P_{kk}\!=\!I$,
the desired signal power can be expressed as
$\frac{g_{kk}\!\cos^2(\beta_{kk}\!-\!\alpha_k)I}{g_{jk}\!\cos^2\!\beta_{kk}}$,
which is an increasing function of $\beta_{kk}$, implying that $P_{kk}$
should be maximized. Therefore the power constraint
at BTS $k$ must be binding.
\end{proof}

Maximizing $R^{(\text{bf})}(\mu)$ is a
non-convex problem; however, since there are only six variables to
optimize it can be solved by exhaustive
search optimally or by an iterative approach, in which the power
allocation is optimized with fixed angles, the angles are optimized
with fixed power allocation, and these two procedures are iterated
until $R^{(\text{bf})}(\mu)$ converges. Note that convergence is
guaranteed since $R^{(\text{bf})}(\mu)$ monotonically increases
in each step, and is bounded due to the power limitations.
The iterative approach can reduce the
search complexity; however, optimality cannot be guaranteed although
in our simulations global optimality was always observed.

\begin{figure}
  \center
  \includegraphics[width=0.5\textwidth]{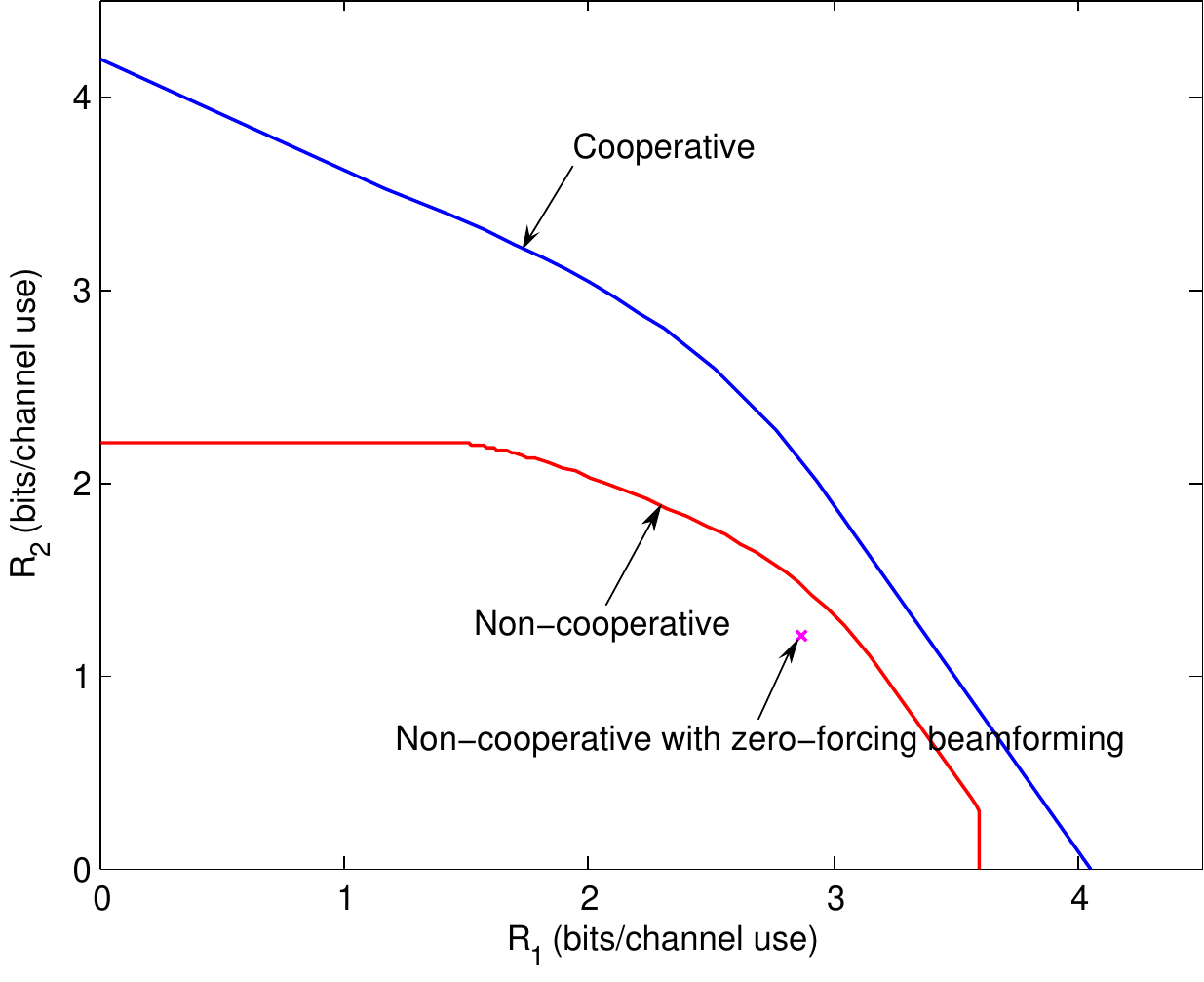}
  \caption{The rate region frontiers achieved by cooperative and non-cooperative schemes with $P_1\!=\!P_2\!=\!3$, $N_t\!=\!2$,
  and randomly generated channels.}
\label{f:beam}
\end{figure}
Fig.~\ref{f:beam} compares the rate region frontiers achieved by
the cooperative scheme with the non-cooperative scheme presented
in~\cite{JorLar08TSP}, in which the BTSs carry out joint beamforming
and power control but do not share messages. Also shown is the rate pair
achieved with zero-forcing transmission at each BTS without cooperation,
i.e., BTS $k$ transmits to its own associated mobile $k$ with the
orthogonal projection $\boldsymbol{I}\!-\!\boldsymbol{\Pi}_{\bar{\boldsymbol{h}}_{jk}},~j\!\neq\!k$.
The figure shows that for this example BTS cooperation gives
substantial gains in $R_2$ when $R_1$ is small.

\subsection{Frequency-Selective Channels}
In a wideband system with frequency-selective channels that are
modeled as a set of $L$ discrete channels, the cooperative
beamforming and power control problem is given by
\begin{subequations}
\label{eq:problem1bf}
\begin{align}
\underset{\{P_{jk}(l),\beta_{jk}(l)\}}{\operatorname{\text{maximize}}}&~~\sum_{l=1}^{L}\big[R_{1}^{(\text{bf})} (l) + \mu R_{2}^{(\text{bf})} (l)\big]\\
\text{subject to}&~~\sum_{l=1}^{L}\big[P_{1k}(l)+P_{2k}(l)\big]=P_{\text{tot},k},~~~P_{jk}(l)\geq0,~\forall j,k,l,\label{eq:beamcon2}
\end{align}
\end{subequations}
where $R_{k}^{(\text{bf})} (l)$ is given by
(\ref{eq:beamRk}), the channel gains $\{g_{jk}\}$,
angles $\{\beta_{jk},\alpha_{k}\}$, and the powers depend on $l$,
and the power constraints in (\ref{eq:beamcon2}) are satisfied with
equality due to Proposition \ref{propbeam}.

As for a single antenna, this is again a two-level optimization
problem. The lower level optimizes
the beamforming vectors to maximize the weighted
sum rate for each sub-channel given the power allocated to each
message. The upper level then optimizes the power allocation across
sub-channels for each message subject to the total power constraints.

Similar to the previous case with a single transmit antenna, the dual problem
associated with Problem (\ref{eq:problem1bf}) can be formulated where
the Lagrangian also depends on the angels $\{\boldsymbol{\beta}_{jk}\}$. (We omit the
details due to space limitations.) As before, letting the number of sub-channels
within a given band tend to infinity, we can assume that
$g_{jk}(l)$ and $\alpha_{k}(l)$ converge to continuous functions
of frequency $f$. The corresponding optimization problem over
$\{P_{jk}(f)\}$ and $\{\beta_{jk}(f)\}$ has zero-duality gap
and can be efficiently solved numerically. The numerical results
in Fig.~\ref{f:m_antenna} were generated by solving the discrete
version in (\ref{eq:problem1bf}) using a nested bisection search
for $\boldsymbol{\lambda}$, where the maximization of the Lagrangian
function in the inner loop is performed over the angles $\{\beta_{jk}(l)\},j,k=1,2$.

\begin{figure}
  \center
  \subfigure[]{
  \includegraphics[width=0.5\textwidth]{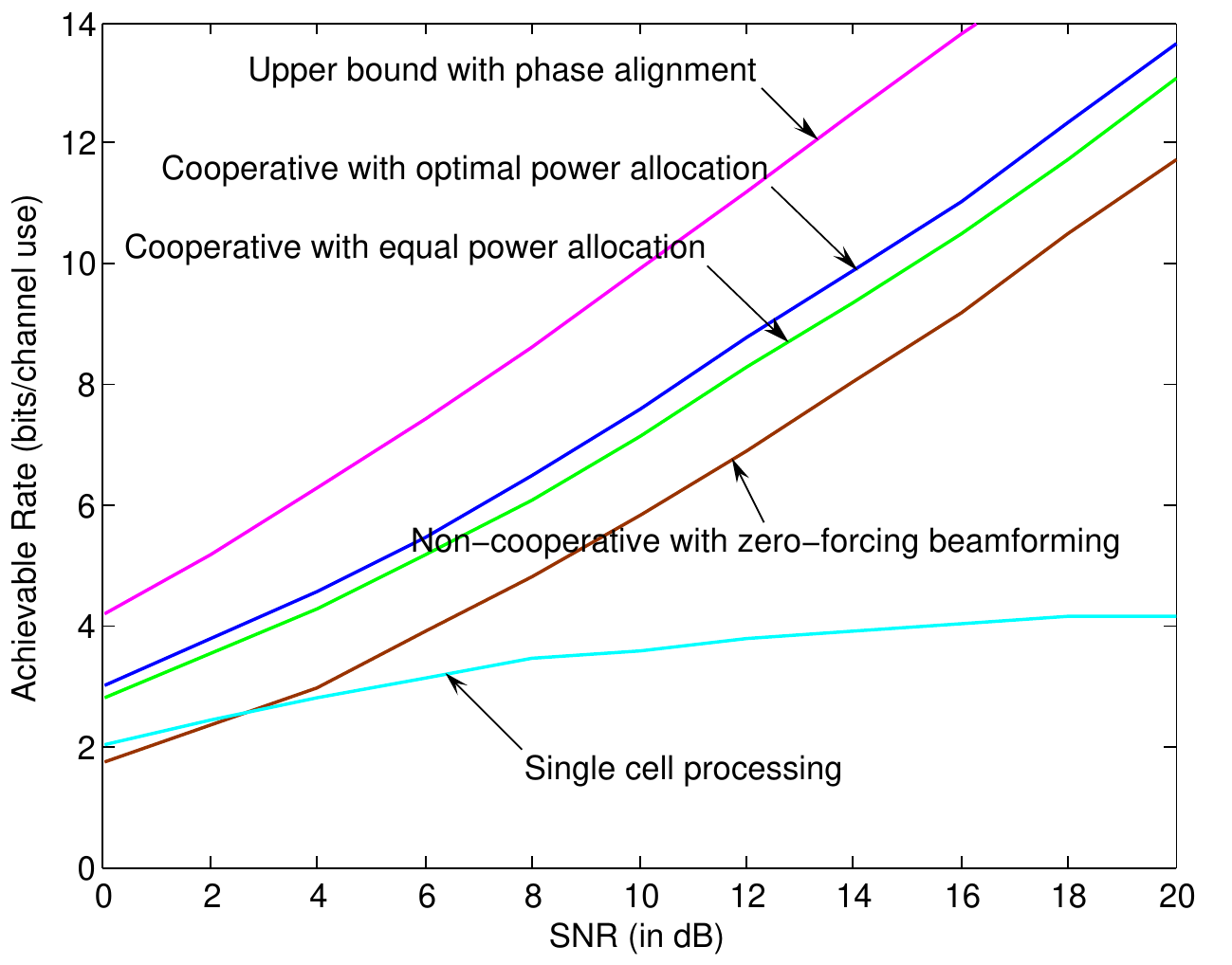}
  \label{f:widebeam}
  }
  \subfigure[]{
  \includegraphics[width=0.5\textwidth]{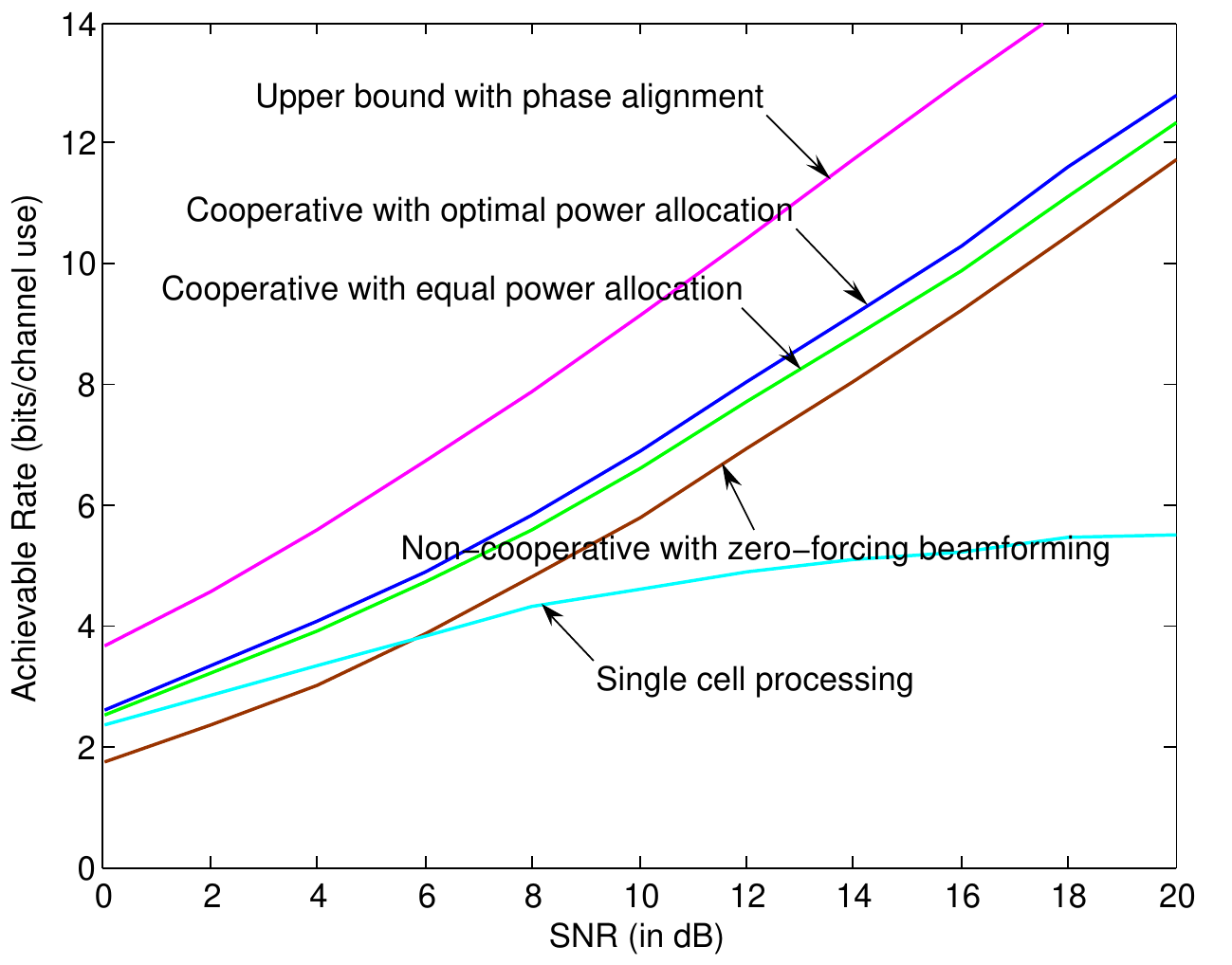}
  \label{f:widebeam_3dB}
  }
  \caption{The achievable sum rate with $N_t\!=\!2$ and (a) equal direct- and cross-channel gains, i.e., $\mathbb{E}[g_{jk}]\!=\!1,~\forall j,k$; (b) direct-channels are 3 dB stronger than cross-channels, i.e., $\mathbb{E}[g_{jj}]\!=\!1,~j\!=\!1,2$ and $\mathbb{E}[g_{jk}]\!=\!0.5,~j,k\!=\!1,2,~j\!\neq\!k$.}
  \label{f:m_antenna}
\end{figure}
In Fig.~\ref{f:m_antenna} we compare the maximum sum rates of
the cooperative and non-cooperative schemes with wideband channels
and $\mu\!=\!1$. There are $L\!=\!128$ sub-carriers. The four channel
vectors on each sub-carrier undergo independent Rayleigh fading, and for each link the channel vectors across sub-carriers are assumed to have correlation coefficients of $0.95$. The figure compares achievable rates
for the following scenarios: 1) optimized power assignments
across sub-carriers and both BTSs according to (\ref{eq:problem1bf}); 2) cooperative
transmission between BTSs with equal power assignments across
sub-carriers; 3)
joint beamforming between BTSs but without message sharing,
in which case each BTS transmits to its associated mobile in
the null space of the cross-channel to the other mobile \cite{JorLar08TSP}. The achievable downlink
sum rate of perfect BTS cooperation with phase alignment is also included for comparison.
Fig.~\ref{f:widebeam} shows results with equal-variance direct- and
cross-channel gains, and Fig.~\ref{f:widebeam_3dB} shows results
for the case where the cross-channel gains are 3 dB weaker than
the direct-channel gains.

The results in Fig.~\ref{f:widebeam} show that the cooperative
scheme considered offers approximately 4 dB gain relative to the
non-cooperative joint beamforming scheme presented in \cite{JorLar08TSP}.
Also, cooperation with wideband power allocation offers one dB gain
with respect to equal power allocation. The cooperative scheme
achieves the same number of degrees of freedom as with phase alignment
(which is
two since there are two single-antenna mobiles) at the expense of
adding one more antenna at each BTS.\footnote{Note, however, that
with phase alignment two transmit antennas per BTS can support
two additional noninterfering mobiles.}
The performance improvement due to cooperation diminishes
if the average cross-channel gains
become weaker than the direct gains, as illustrated in Fig.~\ref{f:widebeam_3dB}.

\section{Conclusions}
We have presented a two-cell cooperation scheme with message
sharing between BTSs, which does not require transmissions
to be phase-aligned.  With a single antenna at the BTSs and mobiles,
the rates are maximized by optimizing the power allocation across
messages, and also sub-channels in the wideband scenario. The scheme
provides large gains with respect to non-cooperative (single cell)
power optimization, but gains with respect to cooperative power
allocation across the two BTSs without message sharing are relatively
modest, although they can be significant, especially at low SNRs.
The gains are primarily due to cell selection, so that they are most
pronounced when the cross-channel gains are comparable with direct-channel
gains. We also extended our results to cooperative joint beamforming
with message sharing. This can provide more degrees of freedom compared
to the single transmit antenna case, but the gains due to message
sharing are again relatively modest. Finally, the absence of phase
alignment, as assumed here, reduces the achievable degrees of freedom
in the high-SNR regime relative to perfect phase alignment.
Fundamental limits (e.g., achievable rate region or degrees of freedom)
of the broadcast channel considered without
transmitter phase alignment, as well as schemes that exploit
partial phase information are left for future work.

\section*{Appendix: Proof of Lemma 1}
We rewrite (\ref{eq:Rk}) as
\begin{align}
&(2^{R_1}-1)(1+g_{11}P_{21}+g_{12}P_{22})=g_{11}P_{11}+g_{12}P_{12}\label{eq:R1v2}\\
&(2^{R_2}-1)(1+g_{21}P_{11}+g_{22}P_{12})=g_{21}P_{21}+g_{22}P_{22}.\label{eq:R2v2}
\end{align}
First, we consider the case $(2^{R_1}\!-\!1)(2^{R_2}\!-\!1)\!<\!1$. We will verify
that $P_{11}\!+\!P_{21}\!=\!P_1$ and $P_{12}\!+\!P_{22}\!=\!P_2$ by contradiction. Suppose $P_{12}\!+\!P_{22}\!<\!P_2$. Then we can
choose two small positive numbers $\Delta_{P_{12}}$ and
$\Delta_{P_{22}}$ that satisfy
\begin{equation}
\Delta _{P_{12}}=(2^{R_1}-1)\Delta_{P_{22}}\label{eq:deltaP}
\end{equation}
and
\begin{equation*}
(P_{12}+\Delta_{P_{12}})+(P_{22}+\Delta_{P_{22}})\leq P_2.
\end{equation*}
Let $P_{12}^{\prime}\!=\!P_{12}\!+\!\Delta_{P_{12}}$ and
$P_{22}^{\prime}\!=\!P_{22}\!+\!\Delta_{P_{22}}$. From (\ref{eq:deltaP}), if we
replace $P_{12}$ and $P_{22}$ in (\ref{eq:R1v2}) with
$P_{12}^{\prime}$ and $P_{22}^{\prime}$, then the equality still
holds, i.e.,
\begin{equation}
(2^{R_1}-1)(1+g_{11}P_{21}+g_{12}P_{22}^{\prime})=g_{11}P_{11}+g_{12}P_{12}^{\prime}.\label{eq:R1v3}\\
\end{equation}
However, since $(2^{R_2}\!-\!1)\Delta_{P_{12}}\!=\!(2^{R_2}\!-\!1)(2^{R_1}\!-\!1)\Delta_{P_{22}}\!<\!\Delta_{P_{22}}$, combining with (\ref{eq:R2v2}) gives
\begin{equation}
(2^{R_2}-1)(1+g_{21}P_{11}+g_{22}P_{12}^{\prime})<g_{21}P_{21}+g_{22}P_{22}^{\prime}.\label{eq:R2v3}
\end{equation}

Defining the achievable rate pair with the new power
allocation $(P_{11},P_{21},P_{12}^{\prime},P_{22}^{\prime})$ as
$(R_1^{\prime},R_2^{\prime})$, (\ref{eq:R1v3}) and (\ref{eq:R2v3})
imply
\begin{align*}
&R_1^{\prime}=\log\bigg(1+\frac{g_{11}P_{11}+g_{12}P_{12}^{\prime}}{1+g_{11}P_{21}+g_{12}P_{22}^{\prime}}\bigg)=R_1\\
&R_2^{\prime}=\log\bigg(1+\frac{g_{21}P_{21}+g_{22}P_{22}^{\prime}}{1+g_{21}P_{11}+g_{22}P_{12}^{\prime}}\bigg)>R_2.
\end{align*}
This contradicts the assumption that $(R_1,R_2)$ is on the rate
region frontier. Hence $P_{12}\!+\!P_{22}\!=\!P_2$ must hold at the optimum.
Similarly, it can be shown that
$P_{11}\!+\!P_{21}\!=\!P_1$.

If $(2^{R_1}\!-\!1)(2^{R_2}\!-\!1)\!=\!1$, then the optimal power
allocation scheme is not unique. Suppose there exists a solution
that satisfies
$P_{12}\!+\!P_{22}\!<\!P_2$, then we can choose $\Delta_{P_{12}}$ and
$\Delta _{P_{22}}$ to satisfy (\ref{eq:deltaP}) and
\begin{equation*}
(P_{12}+\Delta_{P_{12}})+(P_{22}+\Delta_{P_{22}})= P_2.
\end{equation*}
Then by the preceding argument, the new
power allocation $P_{12}^{\prime}\!=\!P_{12}\!+\!\Delta_{P_{12}}$ and
$P_{22}^{\prime}\!=\!P_{22}\!+\!\Delta_{P_{22}}$ must achieve the same rate as
$P_{12}$ and $P_{22}$. As a consequence, there exists an optimal power
allocation in which the power constraints at both BTSs are satisfied
with equality.

We now consider $(2^{R_1}\!-\!1)(2^{R_2}\!-\!1)\!>\!1$ and show that $P_{11}P_{21}\!=\!0$
and $P_{12}P_{22}\!=\!0$ by contradiction. Suppose $P_{12}\!>\!0$ and $P_{22}\!>\!0$.
Then we can choose two small positive numbers
$\Delta_{P_{12}}$ and $\Delta_{P_{22}}$ that satisfy
(\ref{eq:deltaP}) and
\begin{equation*}
P_{12}-\Delta_{P_{12}}\geq0,~P_{22}-\Delta_{P_{22}}\geq0.
\end{equation*}
Let $P_{12}^{\prime}\!=\!P_{12}\!-\!\Delta_{P_{12}}$ and
$P_{22}^{\prime}\!=\!P_{22}\!-\!\Delta_{P_{22}}$. From (\ref{eq:deltaP}), if
we replace $P_{12}$ and $P_{22}$ in (\ref{eq:R1v2}) with
$P_{12}^{\prime}$ and $P_{22}^{\prime}$ respectively, then (\ref{eq:R1v3})
still holds. In addition, since
$(2^{R_2}\!-\!1)\Delta_{P_{12}}\!=\!(2^{R_2}\!-\!1)(2^{R_1}\!-\!1)\Delta_{P_{22}}\!>\!\Delta_{P_{22}}$,
combining with (\ref{eq:R2v2}), we again obtain (\ref{eq:R2v3}).
As before, if we define the achievable rate pair with the new
power allocation $(P_{11},P_{21},P_{12}^{\prime},P_{22}^{\prime})$
as $(R_1^{\prime},R_2^{\prime})$, (\ref{eq:R1v3}) and
(\ref{eq:R2v3}) imply $R_1^{\prime}\!=\!R_1$ and $R_2^{\prime}\!>\!R_2$,
which contradicts the assumption that $(R_1,R_2)$ is on the rate
region frontier. Therefore, $P_{12}P_{22}\!=\!0$ and by similar arguments
we have $P_{11}P_{21}\!=\!0$.




%

\end{document}